\documentclass[12pt,reqno,fleqn]{amsart}
\usepackage{latexsym}
\usepackage{amssymb}
\usepackage{amsxtra}
\usepackage{amsmath}
\usepackage{amsfonts}
\usepackage{CJK}
\newcommand{\cleqn}{\setcounter{equation}{0}}
\newcommand{\clth}{\setcounter{theorem}{0}}
\newcommand {\sectionnew}[1]{\section{#1}\cleqn\clth}
\newtheorem{theorem}{Theorem}[section]
\newtheorem{lemma}[theorem]{Lemma}

\newtheorem{proposition}[theorem]{Proposition}

\newtheorem{remark}[theorem]{Remark}
\newtheorem{definition}[theorem]{Definition}

\renewcommand{\P}{\mathcal{P}}

\renewcommand{\L}{\mathcal{L}}

\def\({\left(}
\def\){\right)}

\def\[{\begin{eqnarray}}
\def\]{\end{eqnarray}}
\def\d{\partial}

\def\La{\Lambda}

\begin{document}

\title{Constrained lattice-field hierarchies and Toda system with Block symmetry}

\author{
Chuanzhong Li\  } \dedicatory {  Department of Mathematics,  Ningbo University, Ningbo, 315211, China\\
 Email: lichuanzhong@nbu.edu.cn}

\date{}

\maketitle

\begin{abstract}
In this paper, we construct the additional $W$-symmetry and ghost symmetry of two-lattice field integrable hierarchies. Using the symmetry constraint, we construct  constrained two-lattice integrable systems which contain several new integrable difference equations. Under a further reduction, the constrained two-lattice integrable systems can be combined into one single integrable system, namely the well-known one dimensional original Toda hierarchy. We prove that the one dimensional original Toda hierarchy has a nice Block Lie symmetry.
\end{abstract}


\allowdisplaybreaks
 \setcounter{section}{0}

\sectionnew{Introduction}

The most fundamental integrable models in mathematical physics are the KP system and Toda system.
As one of the most important sub-hierarchies of the KP hierarchy by considering
reductions  on the Lax operator, the
constrained KP hierarchy contains a large number of interesting
 soliton equations under the
so-called symmetry constraint\cite{kss,chengyiliyishenpla1991,
chengyiliyishenjmp1995}. One of the constraints means that the negative part of the Lax operator of the
constrained KP hierarchy
is a generator of the additional symmetries of the KP hierarchy.
The Toda hierarchy is  a completely integrable system which has many important applications in mathematics and physics including the theory of Lie algebras  and so on \cite{UT83,UT}. The Toda system has many kinds of reductions or extensions, for example extended Toda hierarchy\cite{CDZ}, bigraded Toda hierarchy \cite{C}-\cite{ourBlock}  which governs the Gromov-Witten invariant of $CP^1$ and orbifords. In this paper, we will construct two constrained
 lattice-field hierarchies which are similar to two dimensional Toda system but without crossing flows.

Additional symmetries  of the KP hierarchy were given by Orlov and
Shulman \cite{os1} through the Orlov-Shulman operator $M$, which can be used
to form a centerless  $W$ algebra. The generating function of additional $W$-symmetries  of the KP type integrable hierarchy constitutes a squared eigenfunction symmetry or ghost symmetry in terms of wave functions\cite{Aratyn,Jipeng,ghostdisc}.
Because of the universality of additional symmetries for integrable systems, in this paper, we will use the additional symmetry of  two-lattice field hierarchies to do a further reduction. The reduction produces constrained lattice-field hierarchies whose further reduction is the well-known original Toda hierarchy.
The infinite dimensional Lie algebra of Block type is
a generalization of the well-known Virasoro algebra and has
been studied intensively for example in literatures \cite{Block,Su}.  Later we provide  this kind of Block type algebraic structure for the bigraded Toda hierarchy \cite{ourBlock}, dispersionless bigraded Toda
 hierarchy \cite{dispBTH} and D type Drinfeld-Sokolov hierarchy\cite{2bkpds}. In this paper, we will further prove that the  reduced Toda hierarchy recovers the original Toda hierarchy with its additional Block Lie algebra.

\sectionnew{Two-lattice field hierarchies}
         Two-lattice field hierarchies considered in this section are two families of evolution equations depending on
infinitely many variables $x=(x_1,x_2,\cdots),y=(y_1,y_2,\cdots)$ respectively and a difference variable $n$. Basing on the discrete KP hierarchy in \cite{Kupershmidt} and Toda hierarchy in \cite{UT83,UT}, now we consider two-lattice field hierarchies as
\begin{equation}
\frac{\partial L}{\partial x_i}=[B_i, L],\ \ \ B_i:=(L^i)_+,
\end{equation}
\begin{equation}
\frac{\partial \bar L}{\partial y_i}=[\bar B_i, \bar L],\ \ \ \bar B_i:=(\bar L^i)_+,
\end{equation}
where $L,\bar L$ are two general pseudo-shift operators
\begin{equation} \label{laxoperatordkptwo}
L(n)=\Lambda + \sum_{j=0}^{\infty} u_j(n)\Lambda^{-j},
\end{equation}
\begin{equation} \label{laxoperatordkpbar}
\bar L(n)=v_{-1}\Lambda^{-1} + \sum_{j=0}^{\infty} v_j(n)\Lambda^{j}.
\end{equation}
These two-lattice field hierarchies can be treated as one part of the two dimensional Toda lattice hierarchy\cite{UT} without crossing flow equations.
Here the shift operator $\Lambda$ acts on a discrete function $f(n)$ as $\Lambda f(n)=f(n+1).$
Similar to the two dimensional Toda hierarchy\cite{UT}, $L$ and $\bar L$ can also be  dressed by dressing operators $\P$ and $\bar \P$
\[\P(n;x)=1+\sum^\infty_{j=1}a_j(n;x)\Lambda^{-j},
\]
\[\bar \P(n;y)=\bar a_0 +\sum^\infty_{j=1}\bar a_j(n;y)\Lambda^{j},
\]
by
\begin{equation}\label{dressing}
L=\P\circ\Lambda\circ \P^{-1},\ \ \ \bar L=\bar \P\circ\Lambda^{-1}\circ \bar \P^{-1}.
\end{equation}

Define an anti-evolution operator $*$ by: $\Lambda^*=\Lambda^{-1},\ f(n)^*=f(n)$ for an arbitrary function $f(n)$. Also we define two wave functions (adjoint wave functions) $w,\bar w(w^*,\bar w^*)$ as following
\begin{equation}
w=\P e^{\xi(n,x,z)},\ \ w^*=\P^{-1*}
e^{-\xi(n,x,z)},
\end{equation}
\begin{equation}\bar w=\bar \P e^{\bar \xi(n,y,z)},\ \ \bar w^*=\bar \P^{-1*}
e^{-\bar \xi(n,y,z)},
\end{equation}
 where,
 \begin{eqnarray}
 &&\xi(n,x,z)=\sum_{m\geq 0}z^mx_m+n \ln z,\ \ \bar \xi(n,y,z)=-\sum_{m\geq 0}z^{-m}y_m+n \ln z.
\end{eqnarray}

One can find the wave functions $w,\bar w$ and their adjoint wave functions $w^*,\bar w^*$ satisfy the following flow equations
\[\label{satowl}\frac{\partial w}{\partial x_i}&=&L^n_+ w,\ \ \frac{\partial w^*}{\partial x_i}=-L^{n*}_+ w^*,\\
\label{satowr}
\frac{\partial \bar w}{\partial y_i}&=&-\bar L^n_- \bar w,\ \ \frac{\partial \bar w^*}{\partial y_i}=\bar L^{n*}_- \bar w^*.\]

Based on the above dressing structures in eq.\eqref{dressing}, two independent additional symmetries will be given in the next section which will be used to construct the constrained system by symmetry constraints.

\sectionnew{$W$-symmetry and ghost symmetry}

In order to give the additional symmetries of two-lattice field hierarchies, similarly as \cite{asv2}, we can define the Orlov-Schulman's $M_L$, $M_R$
 operators
 by
\begin{eqnarray}\label{Moperator}
&&M_L=\P\Gamma_L \P^{-1}, \ \ \ \ \ \ \ M_R=\bar \P\Gamma_R
\bar \P^{-1},
\end{eqnarray}
where
\begin{eqnarray}
 && \Gamma_L=
n\Lambda^{-1}+\sum_{i\geq 0}(i+1 )
\Lambda^ix_i,\ \ \
\Gamma_R=-n\Lambda-\sum_{i\geq 0}(i+1 )
\Lambda^{-i}y_i.
\end{eqnarray}

Like the additional symmetry of the Toda hierarchy in \cite{asv2}, we are now to recall the definition of  the additional flows, and then to
prove that they are symmetries, which are called additional
symmetries of two-lattice field hierarchies. The additional flows over two sets of
independent variables $t^*_{m,l},\bar t^*_{m,l}$ and their actions on the wave operators are defined  as
\begin{eqnarray}\label{definitionadditionalflowsonphi2}
\dfrac{\partial \P}{\partial
{t^*_{m,l}}}=-\left(M_L^mL^l\right)_{-}\P, \ \ \ \dfrac{\partial
\bar \P}{\partial {\bar t^*_{m,l}}}=\left(M_R^m\bar L^l\right)_{+}\bar \P,
\end{eqnarray}
where $m\geq 0, l\geq 0$.
The additional flows  can be proved to commute
with the flows of the two-lattice field hierarchies, i.e.,
\begin{eqnarray}
[\partial_{t^*_{m,l}}, \partial_{x_{n}}]\Phi=0,\ \ [\partial_{\bar t^*_{m,l}}, \partial_{y_{n}}]\Psi=0,
\end{eqnarray}
where $\Phi$ can be $\P, L$, and $\Psi$ can be $\bar \P$ or $\bar L$,  and
 $
\partial_{t^*_{m,l}}=\frac{\partial}{\partial{t^*_{m,l}}},
\partial_{x_{n}}=\frac{\partial}{\partial{x_{n}}},\ \partial_{y_{n}}=\frac{\partial}{\partial{y_{n}}}$.

The commutative property  means that
additional flows are symmetries of the lattice-field hierarchy.
Since they are symmetries, one can find that the algebraic
structures among these additional symmetries are two separated $W$-algebras which is included in the following known important
proposition.
\begin{proposition}\label{WinfiniteCalgebra}
The additional flows  $\partial_{t^*_{m,l}}$ form a $W$-algebras with the
following relation
 \begin{eqnarray}\label{algebra relation}
[\partial_{t^*_{m,l}},\partial_{t^*_{n,k}}] \P &= C_{m,l,n,k}^{p,q}\partial_{t^*_{p,q}}\P,\ \
[\partial_{\bar t^*_{m,l}},\partial_{\bar t^*_{n,k}}] \bar \P &= C_{m,l,n,k}^{p,q}\partial_{\bar t^*_{p,q}}\bar \P,
\end{eqnarray}
where $C_{m,l,n,k}^{p,q}$ is the coefficients of the $W$-algebra and  $m,n,l,k\geq 0.$
\end{proposition}

Similar as the method in \cite{asv2}, two generating functions of independent additional
symmetries can be constructed as following
\begin{eqnarray}\label{generatingfunction}
Y(\lambda, \mu)
 &=&\sum\limits_{m=0}^{\infty}\dfrac{(\mu-\lambda)^m}{m!}
\sum\limits_{l=-\infty}^{\infty}\lambda^{-l-m-1}
(M_L^mL^{l+m})_{-},
\end{eqnarray}
\begin{eqnarray}\label{generatingfunction*}
\bar Y(\lambda,
\mu)&=&\sum\limits_{m=0}^{\infty}\dfrac{(\mu-\lambda)^m}{m!}
\sum\limits_{l=-\infty}^{\infty}\lambda^{-l-m-1}
(M_R^mL^{l+m})_{+},
\end{eqnarray}
which can be expressed by a simple form in the sequent
proposition. To this end, we need several well known and useful
techniques in the following several lemmas.
Here we define $res_{_{\lambda}} \sum_{i=-\infty}^{+\infty}A_i \lambda^i=A_{-1},\ res_{_{\Lambda}} \sum_{i=-\infty}^{+\infty}A_i \Lambda^i=A_0,$
which will be used in the following lemma.
\begin{lemma}
For two pseudo-shift operators $P$ and $Q$, the identities
\begin{equation}\label{PQidentity}
res_z[(Pe^{\xi(n,x,z)})(Q(n-1)
e^{-\xi(n,x,z)})]=res_{_{\Lambda}}[P Q^*],
\end{equation}
\begin{equation}\label{PQidentity*}
res_z[(Pe^{\bar \xi(n,y,z)})( Q(n-1)
e^{-\bar \xi(n,y,z)})]=res_{_{\Lambda}}[PQ^*]
\end{equation}
 hold true.
\end{lemma}
\begin{proof} We suppose \begin{eqnarray}
P&=&\sum_{i=-\infty}^{\infty}p_i\Lambda^i,\ \ Q=\sum_{j=-\infty}^{\infty}q_j\Lambda^j.
\end{eqnarray}
Then we get
\begin{eqnarray*}\label{PQidentity'}
&&res_z[(Pe^{\xi(n,x,z)})(Q(n-1) e^{-\xi(n,x,z)})]\\
&=&res_z[(\sum_{i=-\infty}^{\infty}p_i\Lambda^ie^{\xi(n,x,z)})(\sum_{j=-\infty}^{\infty}q_j(n-1)\Lambda^j
e^{-\xi(n,x,z)})]\\
&=&res_z[\sum_{i=-\infty}^{\infty}p_iz^{i}\sum_{j=-\infty}^{\infty}q_j(n-1)z^{-j}]\\
&=&\sum_{i-j=-1}p_i(n)q_j(n-1),
\end{eqnarray*}
and
\begin{eqnarray*}\label{PQidentity''}
&&res_{_{\Lambda}}[P Q^*]\\
&=&res_{_{\Lambda}}[\sum_{i=-\infty}^{\infty}p_i\Lambda^i(\sum_{j=-\infty}^{\infty}q_j\Lambda^j)^*]\\
&=&res_{_{\Lambda}}[\sum_{i=-\infty}^{\infty}\sum_{j=-\infty}^{\infty}p_i\Lambda^i\Lambda^{-j}q_j]\\
&=&\sum_{i-j=-1}p_i(n)q_j(n-1),
\end{eqnarray*}
which will finish the proof of the eq.\eqref{PQidentity}.
For the proof of the eq.\eqref{PQidentity*}, one can derive it by similar calculations.

\end{proof}
After this, the following two lemmas on residues can also be frequently used.
\begin{lemma} If $f(z)=\sum\limits_{-\infty}^\infty a_iz^{-i}$, then
\begin{equation}\label{deltafunction}
{\mathrm res}_z[\zeta^{-1}(1-z/\zeta)^{-1}+z^{-1}(1-\zeta/z)^{-1}]
f(z)=f(\zeta).
\end{equation}
 (Here $(1-z/\zeta)^{-1}$ is understood as a
series in $\zeta^{-1}$ while $(1-\zeta/z)^{-1}$ is a series in
$z^{-1}$.)
\end{lemma}
\begin{lemma}
Let $P$ be a pseudo-shift operators $P=\sum p_i\Lambda^i$,
then
\begin{equation}\label{Pnegativepart}
P_{-}=\sum\limits_{i=1}^{\infty}\Lambda^{-i+1}
[res_{_{\Lambda}}(\Lambda^{i-1}P)]\Lambda^{-1},\ \
P_{+}
=\sum\limits_{i=0}^{\infty}\Lambda^{i+1}
[res_{_{\Lambda}}(\Lambda^{-i-1}P)]\Lambda^{-1}.
\end{equation}
\end{lemma}
\begin{proof}
The proof can be derived by the following two direct calculations
\begin{eqnarray*}\label{Pnegativepart'}
P_{-}&=&\sum\limits_{i=1}^{\infty}
p_{-i}\Lambda^{-i}\\
&=&\sum\limits_{i=1}^{\infty}\Lambda^{-i+1}
p_{-i}(n+(i-1))\Lambda^{-1}\\
&=&\sum\limits_{i=1}^{\infty}\Lambda^{-i+1}
[res_{_{\Lambda}}(\Lambda^{i-1}\sum_{j} p_j\Lambda^j)]\Lambda^{-1}\\
&=&\sum\limits_{i=1}^{\infty}\Lambda^{-i+1}
[res_{_{\Lambda}}(\Lambda^{i-1}P)]\Lambda^{-1},
\end{eqnarray*}
\begin{eqnarray*}\label{Pnegativepart'1}
P_{+}&=&\sum\limits_{i=0}^{\infty}
p_{i}\Lambda^{i}\\
&=&\sum\limits_{i=0}^{\infty}\Lambda^{i+1}
p_{i}(n-i-1)\Lambda^{-1}\\
&=&\sum\limits_{i=0}^{\infty}\Lambda^{i+1}
[res_{_{\Lambda}}(\Lambda^{-i-1}\sum_{j} p_j\Lambda^j)]\Lambda^{-1}\\
&=&\sum\limits_{i=0}^{\infty}\Lambda^{i+1}
[res_{_{\Lambda}}(\Lambda^{-i-1}P)]\Lambda^{-1}.
\end{eqnarray*}
\end{proof}

Then basing on above three lemmas, we can get the following important theorem.

\begin{theorem}\label{generatingfunctionadditionsymmetries}
The generating function of additional symmetries $Y(\lambda, \mu),\bar Y(\lambda, \mu)$ have the following simple form
\begin{equation}
Y(\lambda, \mu)=w(n,x,\mu)
\frac{\Lambda^{-1}}{1-\Lambda^{-1}}w^*(n,x,\lambda),
\end{equation}
\begin{equation}
\bar Y(\lambda, \mu)=  \bar w(n,y,\mu)
\frac{1}{1-\Lambda}\bar w^*(n,y,\lambda).
\end{equation}
\end{theorem}
\begin{proof}
 Using the above three
lemmas, we get
\begin{eqnarray*}
\left(M_L^mL^{l+m}
\right)_{-}&=&\sum\limits_{i=1}^{\infty}\Lambda^{-i+1}\
res_{_{\Lambda}}
[\Lambda^{i-1}\ \P\ \Gamma_L^m\ \Lambda^{N(l+m)} \P^{-1}]\Lambda^{-1}\\
&=&\sum\limits_{i=1}^{\infty}\Lambda^{-i+1}\ res_{_{\lambda}} [
\left(\Lambda^{i-1}\ \P\ \Gamma_L^m\ \Lambda^{l+m}
e^{\xi(n,x,\lambda)}\right)
\left((\P^{-1})^*(n-1)e^{-\xi(n,x,\lambda)}\right)]\Lambda^{-1}\\
&=&\sum\limits_{i=1}^{\infty}\Lambda^{-i+1}\ res_{_{\lambda}} [
\left(\ \lambda^{l+m}\Lambda^{i-1}\ \P\ \Gamma_L^m
e^{\xi(n,x,\lambda)}\right)
\left(w^*(n-1,x,\lambda)\right)]\Lambda^{-1}\\
&=&\sum\limits_{i=1}^{\infty}\Lambda^{-i+1}\ res_{_{\lambda}} [
\left(\ \lambda^{l+m}\Lambda^{i-1}M_L^m \P
e^{\xi(n,x,\lambda)})\right)
\left(\Lambda^{-1}w^*(n,x,\lambda)\right)]\Lambda^{-1}\\
&=&\sum\limits_{i=1}^{\infty}\Lambda^{-i+1}\ res_{_{\lambda}} [
\left(\ \lambda^{l+m}\Lambda^{i-1} M_L^m w(x,\lambda)\right)
\left(\Lambda^{-1}w^*(n,x,\lambda)\right)]\Lambda^{-1}\\
&=&\sum\limits_{i=1}^{\infty}\Lambda^{-i+1}\ res_{_{\lambda}} [
\left(\ \lambda^{l+m}(\d_\lambda^m
w(x,\lambda))(n+(i-1),x,\lambda) \right)
\left(w^*(n-1,x,\lambda)\right)]\Lambda^{-1}\\
&=&\sum\limits_{i=1}^{\infty}\ res_{_{\lambda}} [ \left(\
\lambda^{l+m}(\d_\lambda^m w(x,\lambda))(n,x,\lambda) \right)
\left(w^*(n-i ,x,\lambda)\right)]\Lambda^{-i}.
\end{eqnarray*}

Take the above result back into $Y(\lambda, \mu)$, which becomes
\begin{eqnarray*}
&&Y(\lambda,\mu)=\sum\limits_{m=0}^{\infty}\dfrac{(\mu-\lambda)^m}{m!}
\sum\limits_{l=-\infty}^{\infty}\lambda^{-l-m-1}
(M_L^mL^{l+m})_{-}\\
&=&\sum\limits_{m=0}^{\infty}\dfrac{(\mu-\lambda)^m}{m!}
\sum\limits_{l=-\infty}^{\infty}
\lambda^{-l-m-1}\sum\limits_{i=1}^{\infty}\ res_{_{z}} [ \left(\
z^{l+m}(\d_z^m w(x,z)) \right) \left(w^*(n-i
,x,z)\right)]\Lambda^{-i}\\
&=& \sum\limits_{l=-\infty}^{\infty}
\lambda^{-l-m-1}\sum\limits_{i=1}^{\infty}\ res_{_{z}} [
\left(z^{l+m}( w(n,x,z+\mu-\lambda)) \right)
\left(w^*(n-i
,x,z)\right)]\Lambda^{-i}\\
&=&\sum\limits_{l=-\infty}^{\infty}
\frac{z^{l+m}}{\lambda^{l+m+1}}\sum\limits_{i=1}^{\infty}\
res_{_{z}} [ \left(( w(n,x,z+\mu-\lambda)) \right)
\left(w^*(n-i
,x,z)\right)]\Lambda^{-i}\\
&=& \sum\limits_{i=1}^{\infty}\ res_{_{z}}
[(\frac{1}{z(1-\frac{\lambda}{z})}+\frac{1}{\lambda(1-\frac{z}{\lambda})})
w(n,x,z+\mu-\lambda) w^*(n-i ,x,z)]\Lambda^{-i}\\
&=& \sum\limits_{i=1}^{\infty} w(n,x,\mu) w^*(n-i
,x,\lambda)\Lambda^{-i}\\
&=& \sum\limits_{i=1}^{\infty} w(n,x,\mu)
\Lambda^{-i}w^*(n,x,\lambda)\\
&=&  w(n,x,\mu)
\frac{\Lambda^{-1}}{1-\Lambda^{-1}}w^*(n,x,\lambda).
\end{eqnarray*}
Similarly for the other lattice-field hierarchy, the similar calculations can be done as

\begin{eqnarray*}
\left(M_R^mL^{l+m}
\right)_{+}&=&\sum\limits_{i=0}^{\infty}\Lambda^{i+1}\
res_{_{\Lambda}}
[\Lambda^{-i-1}\ \bar \P\ \Gamma_R^m\ \Lambda^{-(l+m)} \bar \P^{-1}]\Lambda^{-1}\\
&=&\sum\limits_{i=0}^{\infty}\Lambda^{i+1}\ res_{_{\lambda}} [
\left(\Lambda^{-i-1}\ \bar \P\ \Gamma_R^m\ \Lambda^{-(l+m)}
e^{\bar \xi(n,y,\lambda)}\right)
\left((\bar \P^{-1})^*e^{\bar \xi(n,y,\lambda)}\right)]\Lambda^{-1}\\
&=&\sum\limits_{i=0}^{\infty}\Lambda^{i+1}\ res_{_{\lambda}} [
\left(\ \lambda^{-l-m}\Lambda^{-i-1}\ \bar \P\ \Gamma_R^m
e^{-\bar \xi(n,y,\lambda)}\right)
\left( \bar w^*(n-1,y,\lambda)\right)]\Lambda^{-1}\\
&=&\sum\limits_{i=0}^{\infty}\Lambda^{i+1}\ res_{_{\lambda}} [
\left(\ \lambda^{-l-m}\Lambda^{-i-1}M_R^m \bar \P
e^{\bar \xi(n,y,\lambda)})\right)
\left(\bar w^*(n-1,y,\lambda)\right)]\Lambda^{-1}\\
&=&\sum\limits_{i=0}^{\infty}\Lambda^{i+1}\ res_{_{\lambda}} [
\left(\ \lambda^{-l-m}(\d_{\lambda^{-1}}^m
\bar w(n-i-1,y,\lambda)) \right)
\left( \bar w^*(n-1,y,\lambda)\right)]\Lambda^{-1}\\
&=&\sum\limits_{i=0}^{\infty}\ res_{_{\lambda}} [ \left(\
\lambda^{-l-m}(\d_{\lambda^{-1}}^m \bar w(y,\lambda))(n,y,\lambda) \right)
\left(\bar w^*(n+i ,y,\lambda)\right)]\Lambda^{i}.
\end{eqnarray*}

Take this back into $\bar Y(\lambda,\mu)$ will lead to
\begin{eqnarray*}
&&\bar Y(\lambda,\mu)=\sum\limits_{m=0}^{\infty}\dfrac{(\mu-\lambda)^m}{m!}
\sum\limits_{l=-\infty}^{\infty}\lambda^{l+m-1}
(M_R^mL^{l+m})_{+}\\
&=&\sum\limits_{m=0}^{\infty}\dfrac{(\mu-\lambda)^m}{m!}
\sum\limits_{l=-\infty}^{\infty}
\lambda^{l+m-1}\sum\limits_{i=0}^{\infty}\ res_{_{z}} [ \left(\
z^{-l-m}(\d_z^m \bar w(y,z)) \right) \left(\bar w^*(n+i
,y,z)\right)]\Lambda^{i}\\
&=& \sum\limits_{l=-\infty}^{\infty}
\lambda^{l+m-1}\sum\limits_{i=0}^{\infty}\ res_{_{z}} [
\left(z^{-l-m}( \bar w(n,y,z+\mu-\lambda)) \right)
\left(\bar w^*(n+i
,y,z)\right)]\Lambda^{i}\\
&=&\sum\limits_{l=-\infty}^{\infty}
\frac{\lambda^{l+m-1}}{z^{l+m}}\sum\limits_{i=0}^{\infty}\
res_{_{z}} [ \left(( \bar w(n,y,z+\mu-\lambda)) \right)
\left(\bar w^*(n+i
,y,z)\right)]\Lambda^{i}\\
&=& \sum\limits_{i=0}^{\infty}\ res_{_{z}}
[(\frac{1}{z(1-\frac{\lambda}{z})}+\frac{1}{\lambda(1-\frac{z}{\lambda})})
\bar w(n,y,z+\mu-\lambda) \bar w^*(n+i ,y,z)]\Lambda^{i}\\
&=& \sum\limits_{i=0}^{\infty} \bar w(n,y,\mu) \bar w^*(n+i
,y,\lambda)\Lambda^{i}\\
&=& \sum\limits_{i=0}^{\infty} \bar w(n,y,\mu)
\Lambda^{i}\bar w^*(n,y,\lambda)\\
&=&  \bar w(n,y,\mu)
\frac{1}{1-\Lambda}\bar w^*(n,y,\lambda),
\end{eqnarray*}
where the eq.(\ref{deltafunction}) is used.
\end{proof}
The generating functions tell us the lattice-field hierarchies have the following symmetry which is sometimes called ``ghost symmetry".
\begin{proposition}The  two-lattice field hierarchies have the following additional ghost symmetry
\begin{equation}
\d_t\P(n,x)=-(w(n,x)
\frac{\Lambda^{-1}}{1-\Lambda^{-1}}w^*(n,x))\P(n,x),
\end{equation}
\begin{equation}
\d_{\bar t}\bar \P(n,y)=  (\bar w(n,y)
\frac{1}{1-\Lambda}\bar w^*(n,y))\bar \P(n,y).
\end{equation}
\end{proposition}
\begin{proof}
One can prove the commutativity $[\d_{ x_n},\d_{ t}]\P(n,x)=[\d_{ y_n},\d_{\bar t}]\bar \P(n,y)=0$, because $Y(\lambda, \mu)$ and $\bar Y(\lambda, \mu)$ are the generating functions of additional symmetries of two independent lattice-field hierarchies.
\end{proof}
Then by acting on $e^{\xi}$ and $e^{\bar \xi}$ respectively one can derive the following flow equations about
wave functions $w,w^*,\bar w,\bar w^*$,
\begin{align}\label{ghosteq1}\left\{
\begin{aligned}
\d_tw(n,x)&=-(w(n,x)
\frac{\Lambda^{-1}}{1-\Lambda^{-1}}w^*(n,x))w(n,x),\\
\d_tw^*(n,x)&=(w^*(n,x)
\frac{\Lambda}{1-\Lambda}w(n,x))w(n,x),\end{aligned}\right.
\end{align}

\begin{align}\label{ghosteq2}\left\{
\begin{aligned}
\d_{\bar t}\bar w(n,y)&=  (\bar w(n,y)
\frac{1}{1-\Lambda}\bar w^*(n,y))\bar w(n,y),\\
\d_{\bar t}\bar w^*(n,y)&= - (\bar w^*(n,y)
\frac{1}{1-\Lambda^{-1}}\bar w(n,y))\bar w^*(n,y).\end{aligned}\right.
\end{align}

The eq.\eqref{ghosteq1} and eq.\eqref{ghosteq2} are two new integrable coupled nonlocal difference equations to our best knowledge.

\sectionnew{Constrained two-lattice field hierarchies}

Using the above ghost symmetry, we can do a reduction over the Lax operators of two-lattice field hierarchies, i.e.
the Lax operators of the constrained lattice-field hierarchies can be reduced into the following operators
\begin{equation}\label{constraintL}
\L=\Lambda+u+w(n,x)
\frac{\Lambda^{-1}}{1-\Lambda^{-1}}w^*(n,x),
\end{equation}
\begin{equation}\label{constraintbarL}
\bar \L= v\Lambda^{-1}+ \bar w(n,y)
\frac{1}{1-\Lambda}\bar w^*(n,y).
\end{equation}
Then the wave function under the above constraints eq.\eqref{constraintL} and eq.\eqref{constraintbarL}  will compatible with the following Sato equations
\begin{equation}\label{satocons}
\d_{x_n}w=\L^n_+w,\ \ \ \d_{x_n}w^*=-(\L^{n}_+)^*w^*,
\end{equation}
\begin{equation}\label{satocons2}
\d_{y_n}\bar w=-\bar \L^n_-\bar w,\ \ \ \d_{y_n}\bar w^*=(\bar \L^{n}_-)^*\bar w^*.
\end{equation}
The two constrained lattice-field hierarchies can also be defined as

\begin{equation}\label{laxconstraind}
\d_{x_n}\L=[\L^n_+,\L],\ \
\d_{y_n}\bar \L=[-\bar \L^n_-,\bar \L].
\end{equation}
Later we will prove that the Lax equations in the eq.\eqref{laxconstraind} is compatible with two Lax operators \eqref{constraintL} and \eqref{constraintbarL} using  the following proposition.

\begin{proposition}\label{+-}
An operator $B:=\sum_{n=0}^{\infty}b_n\La^n$ is a non-negative shift operator, an operator  $C:=\sum_{n=1}^{\infty}c_n\La^{-n}$ is a negative shift operator and $f(n), g(n)$ are two functions of discrete parameter $n$. The following identities hold
\begin{equation}
(Bf \frac{\Lambda^{-1}}{1-\Lambda^{-1}} g)_-=B(f) \frac{\Lambda^{-1}}{1-\Lambda^{-1}} g,\ \ \ (f \frac{\Lambda^{-1}}{1-\Lambda^{-1}} gB)_-=f \frac{\Lambda^{-1}}{1-\Lambda^{-1}}B^*(g),
\end{equation}
\begin{equation}
(Cf \frac{1}{1-\Lambda} g)_+=C(f)\frac{1}{1-\Lambda} g,\ \ \ (f \frac{1}{1-\Lambda} gC)_+=f \frac{1}{1-\Lambda}C^*(g).
\end{equation}
\end{proposition}

Basing on the Proposition \ref{+-},  two constrained Lax operators \eqref{constraintL} and \eqref{constraintbarL} of the two-lattice field constrained hierarchies are compatible with its corresponding Lax equations. This is included in the following proposition.
\begin{proposition}
The Lax equations in the eq.\eqref{laxconstraind} of the two-lattice field constrained hierarchies are compatible with the reduction on two Lax equations \eqref{constraintL} and \eqref{constraintbarL}.

\end{proposition}

\begin{proof}
Using Sato equations \eqref{satowl} and \eqref{satowr}, the following calculations directly lead to the compatibility with the projections of both sides of eqs.\eqref{laxconstraind}
\[\d_{x_n}\L(n,x)_-&=&\d_{x_n}(w(n,x)
\frac{\Lambda^{-1}}{1-\Lambda^{-1}}w^*(n,x))\\
&=&(B_nw(n,x))
\frac{\Lambda^{-1}}{1-\Lambda^{-1}}w^*(n,x)\\
&&-w(n,x)
\frac{\Lambda^{-1}}{1-\Lambda^{-1}}(B_n^*w^*(n,x))\\
&=&[B_n,w(n,x)
\frac{\Lambda^{-1}}{1-\Lambda^{-1}}w^*(n,x)]_-\\
&=&[B_n,\L]_-,\]

\[\d_{y_n}\bar \L(n,y)_+&=&\d_{y_n}(\bar w(n,y)
\frac{1}{1-\Lambda}\bar w^*(n,y))\\
&=&(-\bar \L^n_{-}\bar w(n,y))
\frac{1}{1-\Lambda}\bar w^*(n,y)\\
&&+\bar w(n,y)
\frac{1}{1-\Lambda}(\bar \L^{n*}_{-}\bar w^*(n,y))\\
&=&[-\bar \L^n_{-},\bar w(n,y)
\frac{1}{1-\Lambda}\bar w^*(n,y)]_+\\
&=&[-\bar \L^n_{-},\bar \L]_+.\]
Then the compatibility can be seen easily.
\end{proof}
Comparing with the two dimensional Toda lattice hierarchy, we need to do the following remark on the lattice-field equations.
\begin{remark}
There is an obstacle to define the corresponding constrained two dimensional Toda lattice hierarchy by symmetry constraints eq.\eqref{constraintL} and eq.\eqref{constraintbarL} because one can not define the crossing derivatives of $x_n$ and $y_n$ because of a contradiction with Sato equations. In another word, there does not exist a corresponding constrained two-dimensional Toda hierarchy defined by symmetry constraints eq.\eqref{constraintL} and eq.\eqref{constraintbarL} similarly as the constrained KP hierarchy \cite{chengyiliyishenpla1991,chengyiliyishenjmp1995}.
\end{remark}

Basing on the symmetry constraint of the hierarchy, we choose the first constrained flow equations in  equations \eqref{satocons} and \eqref{satocons2} as an example here.

\subsection{Constrained lattice-field equations}
The $x_1,y_1$ flow equations in  equations \eqref{satocons} and \eqref{satocons2} are the following three-components coupled systems
\begin{align}\left\{
\begin{aligned}
\d_{x_1}w(n)&=w(n+1)+u(n)w(n),\\
\ \d_{x_1}w^*(n)&=-w^*(n-1)-u(n)w^*(n),\ \\
\d_{x_1}u(n)&=w(n+1)w^*(n)-w(n)w^*(n-1),\end{aligned}\right.
\end{align}
\begin{align}\left\{
\begin{aligned}
\d_{y_1}\bar w(n)&=v(n)\bar w(n-1),\\
\d_{y_1}\bar w^*(n)&=-v(n+1)\bar w^*(n+1),\ \\
 \d_{y_1}v(n)&=v(n)(\bar w(n)\bar w^*(n)-\bar w(n-1)\bar w^*(n-1)).\end{aligned}\right.
\end{align}
The $x_2,y_2$ flow equations are as
\begin{align}\left\{
\begin{aligned}\notag
\d_{x_2}w(n)&=[\La^2+(u(n)+u(n+1))\La +u^2(n)+w(n)w^*(n-1)+w(n+1)w^*(n)]w(n),\\ \notag
 \d_{x_2}w^*(n)&=-[\La^{-2}+\La^{-1}(u(n)+u(n+1))+u^2(n)+w(n)w^*(n-1)+w(n+1)w^*(n)]\\
 &\ \ \ \ \ w^*(n),\\ \notag
 \d_{x_2}u(n)&=(\La^{2}-1)w(n)w^*(n-2)+(\La -1)(u(n-1)+u(n))(w(n)w^*(n-1)),\end{aligned}\right.
\end{align}
\begin{align}\left\{
\begin{aligned}
\d_{y_2}\bar w(n)&=-[v(n)v(n-1)\La^{-2}+v(n)(\bar w(n-1)\bar w^*(n-1)+\bar w(n)\bar w^*(n))\La^{-1}]\bar w(n),\\
 \d_{y_2}\bar w^*(n)&=[(\La^{2}v(n)v(n-1)+\La v(n)(\bar w(n-1)\bar w^*(n-1)+\bar w(n)\bar w^*(n)))]\bar w^*(n),\\
 \d_{y_2}v(n)&=[\bar w(n)\bar w^*(n+1)v(n+1)v(n)-v(n)v(n-1)\La^{-2}\bar w(n)\bar w^*(n+1)\\ \notag
 &+ v(n)(\bar w(n-1)\bar w^*(n-1)+\bar w(n)\bar w^*(n))(\bar w(n)\bar w^*(n)-\bar w(n-1)\bar w^*(n-1))]v(n).\end{aligned}\right.
\end{align}
After denoting $w(n),w^*(n),u(n),\bar w(n),\bar w^*(n),v(n)$ as $q_{n},q^*_{n},u_{n},r_{n},r^*_{n},v_{n}$, then the  $x_1,y_1$ flow equations are as
the following simplified form
\begin{align}\left\{
\begin{aligned}
\d_{x_1}q_{n}&=q_{n+1}+u_{n}q_{n},\\
\ \d_{x_1}q^*_{n}&=-q^*_{n-1}-u_nq^*_{n},\ \\
\d_{x_1}u_{n}&=q_{n+1}q^*_{n}-q_{n}q^*_{n-1},\end{aligned}\right.
\end{align}
\begin{align}\left\{
\begin{aligned}
\d_{y_1}r_{n}&=v_{n}r_{n-1},\\
\d_{y_1}r^*_{n}&=-v_{n+1}r^*_{n+1},\ \\
 \d_{y_1}v_{n}&=v_{n}(r_{n}r^*_{n}-r_{n-1}r^*_{n-1}).\end{aligned}\right.
\end{align}
The  $x_2,y_2$ flow equations are as
the following simplified form
\begin{align}\left\{
\begin{aligned}\notag
\d_{x_2}q_n&=q_{n+2}+(u_{n}+u_{n+1})q_{n+1} +u^2_{n}q_{n}+q_{n}^2q^*_{n-1}+q_{n+1}q_{n}q^*_{n},\\ \notag
 \d_{x_2}q^*_{n}&=-q^*_{n-2}-(u_{n-1}+u_{n})q^*_{n-1}-u^2_{n}q^*_{n}-q_{n}q^*_{n-1}q^*_{n}-q_{n+1}q^{*2}_{n},\\ \notag
 \d_{x_2}u_{n}&=q_{n+2}q^*_{n}-q_{n}q^*_{n-2}+(u_{n+1}+u_{n})(q_{n+1}q^*_{n})-(u_{n-1}+u_{n})(q_{n}q^*_{n-1}),\end{aligned}\right.
\end{align}
\begin{align}\left\{
\begin{aligned}
\d_{y_2}r_n&=-v_{n}v_{n-1}r_{n-2}-v_{n}(r_{n-1}r^*_{n-1}+r_{n}r^*_{n})r_{n-1},\\
 \d_{y_2}r^*_n&=v_{n+2}v_{n+1}r^*_{n+2}+ v_{n+1}(r_{n}r^*_{n}+r_{n+1}r^*_{n+1})r^*_{n+1},\\
 \d_{y_2}v_n&=r_{n}r^*_{n+1}v_{n+1}v_{n}^2-v_{n}v_{n-1}r_{n-2}r^*_{n-1}v_{n-2}+ v_{n}^2(r_{n}^2r^{*2}_{n}-r_{n-1}^2r^{*2}_{n-1}).\end{aligned}\right.
\end{align}

\section{Additional Block symmetry of Toda hierarchy}
Here we will do a further reduction over the Lax operators \eqref{constraintL} and  \eqref{constraintbarL} by letting them depend on the same time variables $t_n$ with $t_n=x_n=y_n$. Then one can derive the constraint over the Lax  operator $ \hat \L$ of the Toda hierarchy, i.e.
\[\hat \L=\L=\bar \L=L=\bar L.\]
Under this reduction, we denote reduced dressing operators $\P,\bar \P$ as $S,\bar S$ which have expansions of the form
\begin{gather}
\label{expansion-S}
\begin{aligned}
S&=1+\omega_1(n)\Lambda^{-1}+\omega_2(n)\Lambda^{-2}+\cdots,\\
\bar S&=\bar\omega_0(n)+\bar\omega_1(n)\Lambda+\bar\omega_2(n)\Lambda^{2}+\cdots.
\end{aligned}
\end{gather}

The inverse operators $S^{-1},\bar S^{-1}$ of operators $S,\bar S$ have expansions of the form
\begin{gather}
\begin{aligned}
S^{-1}&=1+\omega'_1(n)\Lambda^{-1}+\omega'_2(n)\Lambda^{-2}+\cdots,\\
\bar S^{-1}&=\bar\omega'_0(n)+\bar\omega'_1(n)\Lambda+\bar\omega'_2(n)\Lambda^{2}+\cdots.
\end{aligned}
\end{gather}

 The Lax  operator $ \hat \L$ of the Toda hierarchy
has the following expansions
\begin{gather}\label{lax expansion}
\begin{aligned}
 \hat \L&=\Lambda+U(n)+V(n)\Lambda^{-1}.
\end{aligned}
\end{gather}
 In fact the Lax  operator $\hat \L$
 can also be equivalently defined by
\begin{align}
\label{two dressing} \hat \L&:=S\circ\Lambda\circ S^{-1}=\bar S\circ\Lambda^{-1}\circ \bar S^{-1}.
\end{align}

In this section we will use  a convenient notation on the operators $\hat B_{j}$ defined as follows
\begin{align}\label{satoS}
\begin{aligned}
\hat B_{j}&:=\frac{\hat \L^j}{j!}.
\end{aligned}
\end{align}

Now  the known definition of the  Toda hierarchy is as following.
\begin{definition}The  Toda hierarchy is a hierarchy in which the dressing operators $S,\bar S$ satisfy the following Sato equations \cite{UT}
\begin{align}
\label{satoSt} \partial_{t_{j}}S&=-(\hat B_{j})_-S,& \partial_{t_{j}}\bar S&=(\hat B_{j})_+\bar S.\end{align}
\end{definition}
 From the previous, one can derive the following well-known Lax equations  of the Toda Hierarchy are as follows \cite{UT}
   \begin{align}
\label{laxtjk}
  \partial_{t_{j}} \hat \L&= [(\hat B_{j})_+,\hat \L].
  \end{align}

We now put the constraint
eq.\eqref{two dressing} into a construction of the flows of additional
symmetries which form the well-known Block algebra.
With the dressing operators given in the eq.\eqref{two dressing}, we introduce Orlov-Schulman operators as following
\begin{eqnarray}\label{Moperator}
&&M=S\Gamma S^{-1}, \ \ \bar M=\bar S\bar \Gamma \bar S^{-1},\ \\
 &&\Gamma=
n\Lambda^{-1}+\sum_{n\geq 0}
(n+1)\Lambda^{n}t_{n},\ \bar \Gamma=
-n\Lambda.
\end{eqnarray}

Then one can prove the Lax operator $\hat \L$ and Orlov-Schulman operators $M,\bar M$ satisfy the following theorem.
\begin{proposition}\label{flowsofM}
The Lax operator $\hat \L$ and Orlov-Schulman operators $M,\bar M$ of the Toda hierarchy
satisfy the following
\begin{eqnarray}
&[ \hat \L,M]=1,\ \ [ \hat \L,\bar M]=1,\\ \label{Mequation}
&\partial_{ t_{n}}M=
[(B_{n})_+,M],\ \ \partial_{ t_{n}}\bar M=[(B_{n})_+,\bar M],\\
&\dfrac{\partial
M^m\hat \L^k}{\partial{t_{n}}}=[(B_{n})_+,
M^m\hat \L^k],\;  \dfrac{\partial
\bar M^m\hat \L^k}{\partial{t_{n}}}=[(B_{n})_+, \bar M^m\hat \L^k].
\end{eqnarray}

\end{proposition}

\begin{proof}
One can prove the eq.\eqref{Mequation} in this proposition by dressing the following two commutative Lie brackets
\begin{eqnarray*}&&[\partial_{ t_{n}}-\frac{\Lambda^{n+1}}{(n+1)!},\Gamma]=0,\ \ [\partial_{ t_{n}},\bar \Gamma]
=0.
\end{eqnarray*}
The other identities can be proved in a similar way.

\end{proof}
We  now  define the additional flows, and will then prove that they are  additional
symmetries of the Toda hierarchy. We introduce additional
independent variables $t_{m,l}$ and define the actions of the
additional flows on the wave operators as
\begin{eqnarray}\label{definitionadditionalflowsonphi2}
\dfrac{\partial S}{\partial
{t_{m,l}}}=-\left((M-\bar M)^m \hat \L^l\right)_{-}S, \ \ \ \dfrac{\partial
\bar S}{\partial {t_{m,l}}}=\left((M-\bar M)^m \hat \L^l\right)_{+}\bar S,
\end{eqnarray}
where $m\geq 0, l\geq 0$.
By performing the derivative on $\hat \L$ dressed by $S$ and
using the additional flow about $S$ in \eqref{definitionadditionalflowsonphi2}, we get
\begin{eqnarray*}
(\partial_{t_{m,l}} \hat \L)&=& (\partial_{t_{m,l}}S)\ \La S^{-1}
+ S\ \La\ (\partial_{t_{m,l}}S^{-1})\\
&=&-((M-\bar M)^m \hat \L^l)_{-} S\ \La\ S^{-1}- S\ \La
S^{-1}\ (\partial_{t_{m,l}}S)
\ S^{-1}\\
&=&-((M-\bar M)^m \hat \L^l)_{-}  \hat \L+  \hat \L ((M-\bar M)^m\hat \L^l)_{-}\\
&=&-[((M-\bar M)^m \hat \L^l)_{-}, \hat \L].
\end{eqnarray*}
Similarly, we perform the derivative on $\hat \L$ dressed by $\bar S$ and
use the additional flow about $\bar S$ in \eqref{definitionadditionalflowsonphi2} to get the following
\begin{eqnarray*}
(\partial_{t_{m,l}}\hat \L)&=& (\partial_{t_{m,l}}\bar S)\ \La \bar S^{-1}
+ \bar S\ \La\ (\partial_{t_{m,l}}\bar S^{-1})\\
&=&((M-\bar M)^m\hat \L^l)_{+} \bar S\ \La^{-1}\ \bar S^{-1}- \bar S\ \La
\bar S^{-1}\ (\partial_{t_{m,l}}\bar S)
\ \bar S^{-1}\\
&=&((M-\bar M)^m\hat \L^l)_{+} \hat \L- \hat \L ((M-\bar M)^m\hat \L^l)_{+}\\
&=&[((M-\bar M)^m\hat \L^l)_{+},\hat \L].
\end{eqnarray*}
Because
\begin{eqnarray}\label{Toda Hierarchyadditionalflow111.}
[M-\bar M,\hat \L]=0,
\end{eqnarray}
therefore one can further derive the following equation
\begin{eqnarray}\label{Toda Hierarchyadditionalflow1111}
\dfrac{\partial \hat \L}{\partial
{t_{m,l}}}=[-\left((M-\bar M)^m\hat \L^l\right)_{-},
\hat \L]=[\left((M-\bar M)^m\hat \L^l\right)_{+}, \hat \L],
\end{eqnarray}
which gives the compatibility of additional flows of  the Toda hierarchy with the reduction condition \eqref{two dressing}.

By the two propositions above, we can prove
the additional flows $\partial_{t_{m,l}}$ commute
with the $\partial_{t_{n}}$flows of the Toda hierarchy, i.e.,
\begin{eqnarray}
[\partial_{t_{m,l}}, \partial_{t_{n}}]\Phi=0,
\end{eqnarray}
where $\Phi$ can be $S$, $\bar S$ or $\hat \L$,  and
 $
\partial_{t_{m,l}}=\frac{\partial}{\partial{t_{m,l}}},
\partial_{t_{n}}=\frac{\partial}{\partial{t_{n}}}$.
The proof is just a direct calculation and standard. One can check the similar proof in \cite{ourBlock},\cite{2bkpds}-\cite{R}.
This commutative property  means that
additional flows are symmetries of the Toda hierarchy.
Since they are symmetries, it is natural to consider the algebraic
structures among these additional symmetries. So we obtain the following important
theorem.
\begin{theorem}\label{WinfiniteCalgebra}
The additional flows  $\partial_{t_{m,l}}$ of the Toda hierarchy form a Block type Lie algebra as
 \begin{eqnarray}\label{algebra relation}
[\partial_{t_{m,l}},\partial_{t_{n,k}}]= (km-n l)\d_{t_{m+n-1,k+l-1}},
\end{eqnarray}
which holds true in the sense of acting on  $S$, $\bar S$ or $\hat \L$ and  $m,n,l,k\geq 0.$
\end{theorem}
\begin{proof}
 By using
 (\ref{definitionadditionalflowsonphi2}), we get
\begin{eqnarray*}
[\partial_{t_{m,l}},\partial_{t_{n,k}}]S&=&
\partial_{t_{m,l}}(\partial_{t_{n,k}}S)-
\partial_{t_{n,k}}(\partial_{t_{m,l}}S)\\
&=&-\partial_{t_{m,l}}\left(((M-\bar M)^n\hat \L^k)_{-}S\right)
+\partial_{t_{n,k}} \left(((M-\bar M)^m\hat \L^l)_{-}S\right)\\
&=&-(\partial_{t_{m,l}}
(M-\bar M)^n\hat \L^k)_{-}S-((M-\bar M)^n\hat \L^k)_{-}(\partial_{t_{m,l}} S)\\
&&+ (\partial_{t_{n,k}} (M-\bar M)^m\hat \L^l)_{-}S+
((M-\bar M)^m\hat \L^l)_{-}(\partial_{t_{n,k}} S).
\end{eqnarray*}
Then by a tedious calculation, one can further get
 \begin{eqnarray*}&&
[\partial_{t_{m,l}},\partial_{t_{n,k}}]S\\
&=&-\Big[\sum_{p=0}^{n-1}
(M-\bar M)^p(\partial_{t_{m,l}}(M-\bar M))(M-\bar M)^{n-p-1}\hat \L^k
+(M-\bar M)^n(\partial_{t_{m,l}}\hat \L^k)\Big]_{-}S\\&&-((M-\bar M)^n\hat \L^k)_{-}(\partial_{t_{m,l}} S)\\
&&+\Big[\sum_{p=0}^{m-1}
(M-\bar M)^p(\partial_{t_{n,k}}(M-\bar M))(M-\bar M)^{m-p-1}\hat \L^l
+(M-\bar M)^m(\partial_{t_{n,k}}\hat \L^l)\Big]_{-}S\\&&+
((M-\bar M)^m\hat \L^l)_{-}(\partial_{t_{n,k}} S)\\
&=&[(nl-km)(M-\bar M)^{m+n-1}\hat \L^{k+l-1}]_-S\\
&=&(km-nl)\d_{t_{m+n-1,k+l-1}}S.
\end{eqnarray*}
Similarly  the same results on $\bar S$ and $\hat \L$ are as follows
 \begin{eqnarray*}
[\partial_{t_{m,l}},\partial_{t_{n,k}}]\bar S
&=&((km-nl)(M-\bar M)^{m+n-1}\hat \L^{k+l-1})_+\bar S\\
&=&(km-nl)\d_{t_{m+n-1,k+l-1}}\bar S,
\\[6pt]
{}[\partial_{t_{m,l}},\partial_{t_{n,k}}]\hat \L&=&
\partial_{t_{m,l}}(\partial_{t_{n,k}}\hat \L)-
\partial_{t_{n,k}}(\partial_{t_{m,l}}\hat \L)\\
&=&[((nl-km)(M-\bar M)^{m+n-1}\hat \L^{k+l-1})_-, \hat \L]\\
&=&(km-nl)\d_{t_{m+n-1,k+l-1}}\hat \L.
\end{eqnarray*}
Now the theorem is proved.
\end{proof}

{\bf {Acknowledgements:}}
   This work is supported by the Zhejiang Provincial Natural Science Foundation of China under Grant No. LY15A010004,  National Natural Science Foundation of China under Grant No. 11201251, 11571192,   the Natural Science Foundation of Ningbo under Grant No.  2015A610157 and K. C. Wong Magna Fund in
Ningbo University.




\begin{thebibliography}{AAA1}
\frenchspacing


\bibitem{kss}B. G. Konopelchenko, J. Sidorenko and W. Strampp,
 $(1+1)$-dimensional integrable systems as
 symmetry constraints of $(2+1)$-dimensional systems,  Phys. Lett. A157(1991), 17-21.

\bibitem{chengyiliyishenpla1991}Y. Cheng and Y. S. Li, The constraint of the Kadomtsev-Petviashvili
equation and its special solutions, Phys. Lett. A157 (1991), 22-26.



\bibitem{chengyiliyishenjmp1995}Y. Cheng, Constraints of the Kadomtsev-Petviashvili hierarchy, J. Math.
Phys.33(1992), 3774-3782.

  \bibitem{UT83}K. Ueno, K. Takasaki, Proc. Japan
Ser. A, V. 59 (1983), P. 167-170, 215-218)

 \bibitem{UT}
 K. Ueno, K. Takasaki, Toda lattice hierarchy.
In \emph{``Group representations and systems of differential
equations'' (Tokyo, 1982)}, 1-95, Adv. Stud. Pure Math., 4,
North-Holland, Amsterdam, 1984.

\bibitem{CDZ}
G. Carlet, B. Dubrovin, Y. Zhang,  The Extended Toda Hierarchy,
Moscow Mathematical Journal  4 (2004), 313-332,.



\bibitem{C}
G. Carlet, The extended bigraded Toda hierarchy, J. Phys. A, 39
(2006), 9411-9435.

\bibitem{ourJMP}
 C. Z. Li, J. S. He, K. Wu, Y. Cheng,  Tau function and  Hirota bilinear equations for the extended  bigraded Toda
 Hierarchy, J. Math. Phys.51(2010),043514.
\bibitem{solutionBTH} C. Z. Li, Solutions of  bigraded Toda hierarchy,
Journal of Physics A 44(2011), 255201.

\bibitem{ourBlock}
 C. Z. Li, J. S. He, Y. C. Su, Block type symmetry of bigraded Toda hierarchy,
J. Math. Phys. 53(2012), 013517.
\bibitem{os1}A. Yu. Orlov, E. I. Schulman, Additional symmetries of integrable equations and
conformal algebra reprensentaion, Lett. Math. Phys. 12(1986),
171-179.


\bibitem{Aratyn} H. Aratyn, E. Nissimov, S. Pacheva, Method of squared eigenfunction potentials
in integrable hierarchies of KP type, Comm. Math. Phys., 193(1998),
493-525.
\bibitem{Jipeng}J. P. Cheng, J. S. He, S. Hu, The ¡°ghost¡± symmetry of the
BKP hierarchy, J. Math. Phys., 51(2010), 053514.
\bibitem{ghostdisc}C. Z. Li, J. P. Cheng, etal. Ghost symmetry of the discrete KP hierarchy, arXiv:1201.4419, to appear in Monatshefte f\"ur
Mathematik.

  \bibitem{Block}
  R. Block, On torsion-free abelian groups and Lie algebras,  Proc. Amer. Math. Soc., 9(1958), 613-620.



  \bibitem{Su}
Y. Su, Quasifinite representations of a Lie algebra of Block type,
J. Algebra 276(2004), 117-128.
 \bibitem{dispBTH}
 C. Z. Li, J. S. He, Dispersionless bigraded Toda hierarchy and its additional symmetry, Reviews in Mathematical Physics, 24(2012), 1230003.





\bibitem{2bkpds}C. Z. Li, J. S. He, Block algebra in two-component BKP and D type Drinfeld-Sokolov hierarchies, J. Math. Phys. 54(2013), 113501.


\bibitem{torus}C. Z. Li, J. S. He, Quantum Torus symmetry of the KP, KdV and BKP hierarchies, Lett. Math. Phys. 104(2014), 1407-1423.

\bibitem{R} C. Z. Li, J. S. He, Supersymmetric BKP systems and their symmetries, Nuclear Physics B 896(2015), 716-737.

\bibitem{asv2}M. Adler, T. Shiota, P. van Moerbeke, A Lax representation for the Vertex operator
and the central extension, Comm. Math. Phys. 171(1995), 547-588.

\bibitem{Kupershmidt}B. A. Kupershmidt, KP or mKP: Noncommutative mathematics of Lagrangian,
Hamiltonian, and integrable systems, Math. Surveys and Monographs 78, Providence, RI: AMS, 2000.

\end{thebibliography}
\end{document}